\documentclass[preprint]{elsarticle}
\usepackage{mathrsfs} 
\usepackage{amssymb}
\usepackage{amsmath}
\usepackage{amsthm}
\usepackage{caption}
\usepackage{subcaption}
\usepackage[noend,ruled,lined, linesnumbered]{algorithm2e}
\newtheorem{theorem}{Theorem}
\newtheorem{lemma}{Lemma}

\newtheorem{observation}{Observation}[section]
\newtheorem{corollary}{Corollary}
\usepackage{xcolor}
\usepackage{multirow}
\theoremstyle{remark}
\journal{Journal}
\begin{document}
\begin{frontmatter}


\title{Efficient D-2-D via Leader Election: Arbitrary Initial Configuration and No Global Knowledge\tnoteref{mytitlenote}}
 
\tnotetext[mytitlenote]{A preliminary version of this article appeared in the $10^{th}$ Annual International Conference on Algorithms and Discrete Applied Mathematics (CALDAM 2024) \cite{GorainKM24}.}

\author[inst1]{Tanvir Kaur}
\author[inst2]{Barun Gorain}
\author[inst1]{Kaushik Mondal\corref{mycorrespondingauthor}}
\cortext[mycorrespondingauthor]{Corresponding author}
\ead{kaushik.mondal@iitrpr.ac.in}

\affiliation[inst1]{
            addressline={Department of Mathematics, Indian Institute of Technology Ropar}, 
            city={Rupnagar},
            postcode={140001}, 
            state={Punjab},
            country={India}}

\affiliation[inst2]{
            addressline={Department of Electrical Engineering and Computer Science, Indian Institute of Technology Bhilai}, 
            postcode={492015}, 
            state={Chattisgarh},
            country={India}}
\begin{abstract}
Distance-2-Dispersion (D-2-D) problem aims to disperse $k$ mobile agents starting from an arbitrary initial configuration on an anonymous port-labeled graph $G$ with $n$ nodes such that no two agents occupy adjacent nodes in the final configuration, though multiple agents may occupy a single node if there is no other empty node whose all adjacent nodes are also empty. In the existing literature, this problem is solved starting from a rooted configuration for $k$  $(\geq 1)$ agents in $O(m\Delta)$ synchronous rounds with a total of $O(\log n)$ memory per agent, where $m$ is the number of edges and $\Delta$ is the maximum degree of the graph.  In this work we study the D-2-D problem using $n$ mobile agents starting from an arbitrary initial configuration. Solving D-2-D with $n$ agents is equivalent to finding a maximal independent set of the graph as size of any maximal independent set must be less than $n$. We solve this problem and terminate in $O(max\{n\log^2 n, m\})$ rounds using $O(\log n)$ memory per agent. The agents do not have any prior knowledge of any graph parameters. During the run of our algorithm, we also solve the leader election problem that elects an agent as a leader in $O(max\{n\log^2 n, m\})$ rounds with $O(\log n)$ bits of memory at each agent without requiring any prior global knowledge.
\end{abstract}

\begin{keyword}
Mobile agents \sep Collaborative dispersion \sep Deterministic algorithm \sep Distance-2-Dispersion \sep Distributed algorithm
\vspace{-0.15cm}
\end{keyword}

\end{frontmatter}
\section{Introduction}
The problem of dispersion is one of the fundamental problems in distributed computing that aims to position a set of $k~(\leq n)$ mobile agents on the nodes of an anonymous port-labeled graph with $n$ nodes and $m$ edges such that in the final configuration each node contains at most one agent settled on it. The agents positioned on the nodes of the graph store information required by the agents to complete the dispersion of the graph. Thus, a question that naturally arises in this context is: ``What if an additional constraint is imposed such as no two adjacent nodes can be occupied by agents?" This gave rise to the problem of Distance-2-Dispersion (D-2-D) recently \cite{kaur2023}. The authors in \cite{kaur2023} consider a constraint on the dispersion problem in the following form: no two adjacent nodes contain agents in the final configuration. If we begin with $k\geq 1$ agents in the initial configuration, some agents can run out of place while satisfying D-2-D conditions. Thus, the authors allowed the following: an unsettled agent can settle at a node that already contains a settled agent, only if the unsettled agent has no other node to settle at, maintaining the added constraint. Since the agents are not settled at all the nodes of the graph, the agents have limited memory, and the nodes are memory-less, the problem becomes challenging and interesting to work on.

The authors in \cite{kaur2023} solve the problem of D-2-D from a rooted initial configuration. They provide an algorithm that requires $O(m\Delta)$ rounds to solve D-2-D on a graph $G$ consisting of $n$ nodes, $m$ edges and maximum degree $\Delta$. They also provide an $\Omega(m\Delta)$ lower bound on the number of rounds to solve D-2-D given the agents have no more than $O(\log n)$ memory. Also, their algorithm terminates without using any global knowledge of any of the graph parameters.
In this paper, we provide an algorithm that solves the problem of D-2-D using $n$ agents starting from an arbitrary initial configuration in $O(\max\{n\log^2n, m\})$ time complexity with the agents having $O(\log n)$ memory. Our algorithm does not require knowledge of any global parameter and achieves termination. It is easy to observe that, solving D-2-D with $n$ agents is equivalent to finding a maximal independent set of the graph as size of any maximal independent set must be less than $n$. Besides this, the technique used in our algorithm also elects a global leader which is  of independent interest in distributed computing.

\section{Related Work}

In this section, we focus on the existing works that solve dispersion, D-2-D, find maximal independent set, or do leader election using multiple agents on arbitrary graphs.
The problem of dispersion is highly related to the problem of D-2-D. It was first introduced by Augustine et al. \cite{John2018} in the year 2018. The problem of dispersion is worked on under various model assumptions\cite{Ajayicdcn2019,AjayAlgo2019,ShintakuSSS2020,Ajayopodis2021,AAMSS18,KMS2020-ICDCN,KMS2020,tamc19,DBS21,MMM21-TCS,PSM21}. In \cite{Ajayopodis2021}, the authors solved the problem of dispersion from arbitrary initial configuration that runs in $O(min\{m,k\Delta\})$ rounds using $\Theta(\log(k+\Delta))$ bits additional memory per agent, ignoring the $O(\log n)$ memory that agents anyhow require to store their ID. The drawback of this work is that the algorithm does not terminate. However, recently Sudo et al. \cite{Sudo} revolutionized the history of the dispersion problem by creatively solving the problem of dispersion for rooted configuration in $O(k\log min\{k,\Delta\})=O(k\log k)$ time using $O(\log\Delta)$ bits of memory per agent. Further, for the arbitrary initial configuration they achieve dispersion in $O(k(\log k)(\log min\{k,\Delta\}))=O(k\log^2k)$ time using $O(\log(k+\Delta))$ bits of memory by the agents. This is the best-known algorithm till date to solve dispersion. 

By imposing an extra constraint on the problem of dispersion, Kaur et al. \cite{kaur2023} recently introduced the problem of Distance-2-Dispersion (D-2-D). The problem states that given a set of $k\geq1$ agents, they need to reposition themselves and settle at some node satisfying the condition that no two agents can settle at adjacent nodes. Moreover, an agent can settle at a node where there is already a settled agent only if no more nodes are left satisfying the first condition. They also provide a lower bound of $\Omega(m\Delta)$ on the number of rounds required by the agents to solve the problem of D-2-D considering agents do not have more than $O(\log n)$ memory. Further, they provide an algorithm for rooted configuration on arbitrary graphs that requires $2\Delta(8m-3n+3)$ rounds using $O(\log n)$ memory per agent. Later, in \cite{GorainKM24} they show the power of using more agents. They assume the number of agents $k=n+1$, where $n$ is the number of nodes in the graph, and improve over time complexity for rooted configuration. They provide an algorithm for the rooted initial configuration on arbitrary graphs in $O(m)$ rounds using $O(\log n)$ memory per agent. For arbitrary initial configuration, their algorithm requires $O(pm)$ rounds, where $p$ is the number of nodes containing agents in the initial configuration. The presence of $k>n$ agents also ensures the formation of a maximal independent set of the graph. In \cite{PattanayakBCM24}, the authors compute a maximal independent set of the graph using mobile agents that are arbitrarily positioned on the graph in $O(n\Delta\log n)$ time and $O(\log n)$ memory with each agent. However, to run their algorithm, the agents require the knowledge of $n$ and $\Delta$. However, in a recent work \cite{leader24}, the authors compute a maximal independent set of a graph using $n$ mobile agents in $O(n\Delta)$ time, and each agent requires $O(n\log n)$ memory to run the algorithm without any prior knowledge of any global parameter. Further, in \cite{leader24}, the authors solve the problem of leader election using $n$ mobile agents. They propose a deterministic algorithm to elect an agent as the leader in $O(m)$ rounds and $O(n\log n)$ bits at each agent. The agents do not require any prior knowledge of any of the global parameters.

Another work that computes a maximal independent set of the graph using mobile agents is \cite{Pramanick23}. Pramanick et al. proposed an algorithm to find a maximal independent set using myopic luminous agents \cite{Pramanick23} of an arbitrary connected graph under the asynchronous scheduler. However, according to their model, the agents have prior knowledge of $\Delta$ and agents have at least $3$ hops visibility. Agents use colors to represent different states as a medium of communication. 
The problem of scattering or uniform distribution is also related to the D-2-D problem. The scattering problem is studied on grids \cite{BarriereFBS09} and on rings \cite{ElorB11,ShibataMOKM16}. Both problems are studied with anonymous agents.

In this work, we investigate the problem of Distance-2-Dispersion (D-2-D) by a group of $n$ agents that are arbitrarily positioned at the nodes of a graph. Each agent has $O(\log n)$ memory and does not have any global knowledge. By the end, the agents achieve D-2-D configuration by placing themselves at the nodes of the graph that constitute a maximal independent set of the graph. Our algorithm runs in $O(max\{n\log^2n,m\})$ time. This is an extension and improvement over the results provided in \cite{GorainKM24}. Further, our technique helps in the election of a global leader in the graph which is an important problem vastly studied in distributed computing.

\section{The Model and the Problem}
\noindent\textbf{Model:} Let $G$ be a simple, connected, undirected graph with $n$ nodes and $m$ edges. The graph is zero-storage i.e. the nodes have no memory. Also, $G$ is anonymous i.e. the nodes of the graph have no ID. However, the edges incident to a node $v$ are locally labeled at $v$ so that an agent present at $v$ can distinguish between those edges. In other words, the graph is port-labeled which implies that each edge associated with any node $v$ has a distinct numbering from the range $[0,\delta(v)-1]$ where $\delta(v)$ is the degree of node $v$. These port numbers associated with both ends of an edge are independent of each other. For any node $v$, we define $N(v, i)$ as the node $u\in N(v)$ such that the port associated with the edge $(v,u)$ at $v$ is $i$. 

We assume that $n$ agents are present arbitrarily at the nodes of the graph. The set of all agents is denoted by $A$. An agent is always present at any node of the graph and is never located at any edge at any time step. Each agent has a unique identifier in the range $[1, n^c]$, for some constant $c$, and has $O(\log n)$ memory. We consider the face-to-face communication model, where the agents present at the same node can communicate with each other. An agent can see the port numbers of the respective edges associated with a node where it is currently present. It also knows the port number used to enter into the current node. The agents have no prior knowledge of the network $G$.   

We assume synchronous rounds where, in each round, an agent communicates (with co-located agents), computes, and moves (or does not move). In one round, an agent can traverse only one edge from its current position during the move. The number of synchronous rounds required from the start till the termination of all the agents is the time complexity of the algorithm. The memory requirement per agent to run the algorithm is the memory complexity of the algorithm.
\\
\\
\noindent{\textbf{The problem (D-2-D with $n$ agents):}} 
Given a set of $n$ agents placed arbitrarily in a port-labeled graph $G$ with $n$ nodes and $m$ edges, the agents need to achieve a configuration by the end of the algorithm where each agent settles at some node satisfying the following two conditions: (i) no two adjacent nodes can be occupied by settled agents, and (ii) an agent can settle in a node where there is already a settled agent only if no other empty node is available satisfying condition (i).

\noindent Note that, solving D-2-D with $n$ agents is equivalent to finding a maximal independent set of the graph as size of any maximal independent set must be less than $n$.

\section{Our contribution:}
We obtain the following results and improvements over the existing works which are also presented in \noindent Table \ref{table:summary}.
\begin{itemize}
    \item We present an algorithm that solves our problem from an arbitrary initial configuration using $n$ mobile agents in $O(\max\{n\log^2 n, m\})$ rounds, provided each agent is equipped with $O(\log n)$ memory and has no prior knowledge of any of the global parameters. Solving D-2-D using $n$ agents is equivalent to finding a maximal independent set of the graph. 
    \item We also solve the problem of leader election, i.e., electing an agent as the leader of those $n$ agents, in $O(\max\{n\log^2 n, m\})$ rounds with each agent having $O(\log n)$ bits of memory. The agents do not require prior information of any global parameter to run our algorithm. 
    \item We reduce the number of agents from $n+1$ to $n$ and improve the time complexity from $O(pm)$ to $O(\max\{n\log^2 n, m\})$ compared to the preliminary version \cite{GorainKM24}, where $p$ denotes the number of nodes having at least one agent in the initial configuration.
    \item Our results improve over \cite{PattanayakBCM24} with respect to the fact that our algorithm does not require any knowledge of any global parameter. Moreover, our time complexity is better than \cite{PattanayakBCM24} whenever $n\log^2 n\in o(m)$. Note that, in this case, our complexity is in $O(m)$ and $m\in o(n\Delta\log n)$ always holds in any graph. Otherwise, if $m\in o(n\log^2 n)$, then our time complexity is better in all such graphs where $\log n\in o(\Delta)$.
    \item We solve leader election using $O(\log n)$ memory per agent which is a major improvement compared to the $O(n\log n)$ memory per agent as in \cite{leader24}. We match their time complexity whenever $n\log^2 n\in o(m)$. Our algorithm as well as the algorithm in \cite{leader24} does not require any prior global knowledge.
\end{itemize}

\noindent The general idea of achieving D-2-D on a graph is to first achieve dispersion on the graph followed by the conversion of this dispersed configuration into a D-2-D configuration. However, in order to execute this idea, the agents must understand that dispersion is achieved on the graph so that it can begin the further part of the algorithm. However, the existing best-known dispersion algorithm \cite{Sudo} does not achieve termination without the knowledge of any global parameters, though it runs in $O(n\log^2n)$ runtime. This is the main reason why \cite{PattanayakBCM24} requires global knowledge or \cite{GorainKM24} requires $n+1$ agents. Here in this work, we use this negative fact as a positive entity. We say that if an agent $a_i$ sees a neighboring node that does not have a settled agent, $a_i$ understands that the dispersion has not been achieved yet. Theoretically, this implies our runtime is still within $O(n\log^2 n)$ only. Hence $a_i$ can safely wait until it sees a settled agent in the neighboring node it is looking for before proceeding further.

\begin{table}[ht!]
\begin{tabular}{|l|l|l|l|l|}
\hline
\textbf{Output} & \textbf{\begin{tabular}[c]{@{}l@{}}No. of \\ agents\end{tabular}} & \textbf{\begin{tabular}[c]{@{}l@{}}Knowledge\\ of Global\\ Parameters\end{tabular}} & \textbf{Time} & \textbf{Memory} \\ \hline
\begin{tabular}[c]{@{}l@{}}MIS (Our work)\\ (D-2-D conf.)\end{tabular} & $n$ & No & $O(\max\{n\log^2 n, m\})$ & $O(\log n)$ \\ \hline
MIS \cite{PattanayakBCM24} & $n$ & $n$, $\Delta$ & $O(n\Delta\log n)$ & $O(\log n)$ \\ \hline
MIS \cite{leader24} & $n$ & No & $O(n\Delta)$ & $O(n\log n)$ \\ \hline
\begin{tabular}[c]{@{}l@{}}MIS \cite{GorainKM24}\\ (D-2-D conf.) \end{tabular} & $n+1$ & No & $O(pm)$ & $O(\log n)$ \\ \hline
\begin{tabular}[c]{@{}l@{}}Leader Election \\ (Our Work)\end{tabular} & $n$ & No & $O(\max\{n\log^2 n, m\})$ & $O(\log n)$ \\ \hline
\begin{tabular}[c]{@{}l@{}}Leader Election \\ \cite{leader24}\end{tabular} & $n$ & No & $O(m)$ & $O(n\log n)$ \\ \hline

\end{tabular}
\caption{Summary of results from existing works and our work. In all the above-mentioned results, the authors consider the agents arbitrarily positioned at the nodes of a graph $G$. Here $n$ is the number of nodes in $G$, $m$ is the number of edges in $G$, $\Delta$ is the maximum degree of a node in $G$, conf. denotes configuration and $p$ is the number of nodes that contain at least one agent in the initial configuration. }\label{table:summary}
\end{table}

\section{The Algorithm}
In this section, we propose an algorithm for achieving D-2-D. Before giving the details of the algorithm, we start by describing the challenges of solving D-2-D using the existing techniques used to solve dispersion. Next, we give a high-level idea of our proposed algorithm that also explains how to overcome the challenges to achieve D-2-D.  

\subsection{Challenges and High-Level Idea}
The intuitive solution to the D-2-D problem is to first achieve dispersion of the graph and then convert the dispersed configuration into the desired D-2-D one. However, the known algorithms to solve the dispersion of the graph do not achieve termination. The algorithms that achieve termination require prior knowledge of global parameters. Thus, it becomes a challenge to efficiently begin the conversion phase that achieves D-2-D from the dispersed phase. The algorithm \cite{GorainKM24} achieves D-2-D with termination but requires high time complexity and uses $n+1$ agents. The algorithm in \cite{PattanayakBCM24} requires knowledge of several global parameters.
The question is whether we can efficiently achieve D-2-D without using any global knowledge and achieve termination as well along with improving the existing literature. We answer this question affirmatively by providing such an algorithm. 

The idea is to begin with the dispersion algorithm as provided in \cite{Sudo} and merge the multiple DFS trees that may arise due to the dispersion algorithm into a single tree. Now, one can not wait for the dispersion to end before starting the merging as the dispersion algorithm is non-terminating. In our algorithm, agents (some designated agents) look for settled agents in their neighborhood and proceed with the merging in parallel when the dispersion is still running. In case some designated agent doesn't see a settled agent in its neighborhood node (whose presence is must for the merging), it can safely wait as it understands that dispersion is not achieved yet and hence the waiting time remains within the runtime of phase 1. We ensure that by the time merging is done on the whole graph, exactly one agent understands that, and hence we can see this as the leader as well. Now this leader agent initiates the transformation of dispersion configuration to a D-2-D configuration using our subroutine of \cite{GorainKM24} over the tree of $n$ nodes that is formed due to the merging.

\subsection{Overview of our Algorithm}
Recently Sudo et al. introduced a novel approach for performing dispersion on a graph, applicable to both rooted and arbitrary initial configurations, as documented in their work \cite{Sudo}. Leveraging settled agents, they execute Depth-First Search (DFS) on the graph, reducing the time for exploration of all the neighbors of a node from $O(\tau)$ to $O(\log\tau)$, where $\tau = \min\{k, \Delta\}$. Here, $k$ denotes the number of agents in the graph. The time complexity, as per their findings, for dispersion from a rooted configuration is $O(k\log k)$, while for an arbitrary configuration, it is $O(k\log^2 k)$. In our work, we employ their algorithm for dispersion from arbitrary initial configurations, integrating it as a subroutine into our approach. We achieve Distance-2-Dispersion (D-2-D) on a graph from any arbitrary initial configuration by using a couple of other subroutines in parallel that help achieve D-2-D from the dispersed configuration. Our objective is to achieve dispersion on the graph and simultaneously merge the different DFS trees that may exist during dispersion. The dispersion algorithm of \cite{Sudo} is called phase $1$ of our algorithm. The algorithm that is run in parallel to phase $1$ for merging the multiple trees formed by the algorithm of \cite{Sudo} into one tree, is the phase $2$ of our algorithm. Finally, the conversion of the dispersed configuration into the D-2-D configuration is executed in phase $3$ of our algorithm. Our algorithm proceeds in iterations, where each iteration consists of $24$ rounds. In particular, if the current round is round $t$ and $t=24\cdot i +r$, then $t$ is the $r^{th}$ round of $(i+1)^{th}$ iteration. Table \ref{table:slot_func} describes the rounds dedicated for each phase in a specific iteration $i$ along with the activities that the agents do during each round.

Agents run phases $1$ and $2$ in parallel since the algorithm used to achieve dispersion in phase $1$ does not achieve termination. Phase $3$ runs only after phases $1$ and $2$ get completed.
Let the agents be initially positioned at $p$ distinct nodes in the graph. If the number of agents at a node in the initial configuration is at least 2, then they are said to be in a group. Otherwise, they are said to be alone at a node. The agent with the smallest ID among its respective group settles at this current node and thus, this node is said to be the root node for its respective DFS traversal. These agents are also said to be the $leading$ agents. All the agents maintain a label that is the same as the ID of the largest agent in their group. In other words, each group now has a unique label. The $leading$ agents are responsible for the execution of phase $2$ of the algorithm in parallel to phase $1$. A $leading$ agent, say $r_l$ moves from a node $u$ to a node $v$ only if an agent $r$ is settled at node $v$ during phase $1$. The agent $r_l$ traverses through the nodes one by one if there are settled agents at those nodes; otherwise, $r_l$ waits at its current node. The $leading$ agent checks the label of the settled agent $r$. If the settled agent's label is larger than the label of the $leading$ agent $r_l$, then $r_l$ stops at the current node itself. However, if the label of the settled agent $r$ is smaller than that of the $leading$ agent, then $r$ changes its label to the same as that of the $leading$ agent's label. The settled agent $r$ thus becomes a part of the $r_l's$ DFS tree. In this way, the $leading$ agent either stops or continues its single DFS tree construction for phase $2$. By the end of this phase, we ensure that a single DFS tree is constructed by the $leading$ agent $r_{l_{max}}$ that has the largest value of label. This is because $r_{l_{max}}$ does not get stuck at any node since its label is the maximum. Note that if $r_{l_{max}}$ finds an empty node, it understands that phase $1$ is still running and waits till that empty node is occupied by some agent. Note that, the waiting time is within the runtime of phase $1$ which is $O(n\log^2 n)$. Phase $2$ may need at most $O(m)$ time after phase $1$ completes. 

Since the agent settled at the root node is out for conducting phase $2$ of the algorithm, and the phases $1$ and $2$ are run in parallel, unsettled agents may consider such root node to be vacant. Thus, its first child (the agent that settled just after the agent settled at the root node in phase $1$) hops to and fro from its original position and through its $parent$ pointer. Moreover, the entire information stored in the variables of the agent settled at the root node is transferred to the dummy variables of its first child. In other words, when the agent settled at the root node executes phase $2$ of the algorithm, its position is virtually occupied by its first child, rescuing the other DFSs that are running phase $1$ from misinterpreting its position as a vacant node\footnote{If a group consists of a single agent at the start, we take care of this case slightly differently.}. 

Phase $3$ runs according to the \cite{GorainKM24}, where the idea is that the agent $r_{l_{max}}$ moves over the tree edges of the DFS tree constructed in the phase $2$ one by one, where the agents settled at the nodes finally decide whether to stay at their position or vacate it, thereby achieving D-2-D in the graph. For the sake of completeness, now we provide an overview of the dispersion algorithm of \cite{Sudo} and the algorithm to convert dispersed configuration to D-2-D configuration of \cite{GorainKM24}.\\
\\
\noindent{\textbf {Dispersion algorithm \cite{Sudo}}:} Let $G$ be a graph where the agents are arbitrarily placed at the nodes. The agents located at the same starting node are viewed as a single group. The agents $a$ with $a.settled=\top$ are termed as the settlers, and the other agents are called explorers. Explorers are classified into two classes: leaders and zombies. An explorer $a$ is called a leader if $a.leader=a.ID$, else it is called a zombie. All the agents are leaders at the start of the execution of the algorithm. In other words, initially all the agents $a$ have $a.leader=a.ID$. A leader may convert itself into a zombie, which eventually becomes a settler, while a zombie can never become a leader, and a settler can never change itself into a zombie or a leader. Another variable, $lvl$ is used to bound the execution time of the dispersion. The variable $lvl(\in \mathbb{N})$ is set to $1$ initially for each agent. The pair $(a.leader, a.lvl)$ serves as the group identifier for phase $1$ of the algorithm. A relationship $\prec$ between any two non-zombies $r_u$ and $r_v$ uses the defined group identifiers in phase $1$ as follows:
 $r_u \prec r_v \iff (r_u.lvl< r_v.lvl)\vee(r_u.lvl=r_v.lvl \wedge r_u.leader<r_v.leader)$. An agent $a$ is said to be weaker than $b$ if $a \prec b$, else $a$ is said to be stronger than $b$. The pair $(a.leader, a.lvl)$ serves as the group identifier for phase $1$. When a leader $a$ decides to settle an agent at a node, it settles the zombie with the minimum ID, say $z_u$, and $z_u$ sets $(z_u.leader, z_u.lvl)=(a.leader, a.lvl)$. We denote this using $\psi$. In particular, let $z_u$ settle at the node $u$. Then $\psi(u)=z_u$. For any leader $a$, we define the territory of $a$ as all the nodes $w\in V$ such that the agent $r_w$ settled at the node $w$ has $r_w.leader=a.ID$ and $r_w.lvl=a.lvl$.
The novel idea they use that reduces runtime compared to existing literature is called probing, which works as follows. The idea is to find an unvisited neighbor (if any) of the current location with the help of settled agents in the neighborhood of the current node. In particular, let $x$ many agents $\{a_1, a_2, \ldots, a_x\}$ be positioned at a node $u$ and $r_u$ be settled at $u$. The agents need to find the empty neighbor of $u$ (if any) so that all the agents can move to it. Thus, the agents $a_1, a_2, \ldots, a_x$ simultaneously move through the ports $p_1, p_2, \ldots, p_x$. If any agent $a_i$ finds an agent settled at the neighboring node, it brings it along with it at $u$. Thus, either the number of agents doubles at $u$ that can further move to neighboring nodes of $u$, or there is an agent that came alone, i.e., there is an empty neighbor. Thus, all the agents move to this new node. Otherwise, if all the agents cannot find any empty neighbor and there is no further port left to be explored, then they backtrack from $u$. 
\\
\\
Let a leader $a$ be performing its DFS and is currently at a node $u_1$. It performs its probing at the node $u_1$ in order to determine a vacant node (if any) in its neighbourhood. If a node that is detected during the probing, has a settled agent with a different group identifier, then that node is considered to be unsettled. As a consequence, $a$ may move forward to a node that is inside another leader's territory. If the node $u$ has a settled agent that belongs to a weaker group, then $a$ incorporates $\psi(u)$ into its own group by providing the identifier $(a.leader, a.lvl)$ to the agent $\psi(u)$. However, if the leader $a$ encounters a stronger leader or a stronger settler, it becomes a zombie and follows the path towards the stronger leader. The agent $a$ stops its execution of the DFS algorithm and is no longer a leader. Thus, it meets with the stronger leader eventually, further helping it to expand its territory. With this, dispersion is achieved on the graph. After the dispersion configuration is achieved, there may be multiple DFS trees corresponding to different labels. Recall that our objective for phase $2$ is to make a single DFS tree.
In the initial configuration, after the execution of the first round, each group of agents elects a leader. This leader, say $a.leader$ is stored as the group identifier for phase $2$ in another variable, $prm\_leader$. The decision for the merging of the groups for settled agents done in phase $2$ of the algorithm is solely based on the variable $prm\_leader$.\\
\\
\noindent{\textbf {Dispersed configuration to D-2-D \cite{GorainKM24}}}: 
The agent, say $r_{st}$ (the agent that successfully completes phase $2$ and constructs its DFS tree), initiates phase $3$ of the algorithm. It begins the traversal of the graph via the tree edges of the DFS tree formed in phase 2. It receives information about the child ports of each settled agent. After receiving the information about the child port, the agent $r_{st}$ traverses through it and continues to traverse via the child port of each settled agent unless it reaches a leaf node (i.e., a settled agent that has no child in the DFS tree). Let the settled agent at this node be $r_u$. The agent $r_u$ traverses through all its ports one by one and checks if there is any settled agent that is permanently settled at its position. If yes, $r_u$ moves to such a permanently settled agent; otherwise, it permanently settles itself at its original position. In this way, the agent $r_{st}$ moves through the tree edge one by one, and the settled agents decide sequentially whether to stay at their position or not. Thus, finally, the D-2-D configuration is attained by the agents. 

Now, we proceed with a detailed description of our algorithm.
\begin{table}[]
\begin{tabular}{|c|l|l|p{4.3cm}|}
\hline
\multicolumn{1}{|l|}{\textbf{Iteration}} & \textbf{Rounds} & \textbf{Phase}                                                            & \textbf{Activity Performed}                  \\ \hline
\multirow{3}{*}{i}                       & 1-14           & \begin{tabular}[c]{@{}l@{}}Phase 1 \\(Dispersion)\end{tabular}            &Runs Phase 1 unless it is over\\ \cline{2-4} 
                                         & 15-18          & \begin{tabular}[c]{@{}l@{}}Phase 2\\ (A single DFS tree)\end{tabular}     & Runs Phase 2 unless it is over               \\ \cline{2-4} 
                                         & 19-24          & \begin{tabular}[c]{@{}l@{}}Phase 3\\ (Dispersion into D-2-D)\end{tabular} & Runs Phase 3 only if Phases 1 and 2 are over \\ \hline
\end{tabular}
\caption{Functions in each slot}
    \label{table:slot_func}
\end{table}

\subsection{Detailed Description of our Algorithm} \label{sec:algorithm}

In this section, we explain in detail the algorithm to solve D-2-D on a graph with an arbitrary initial configuration, given that the agents have no global knowledge of the graph. Table \ref{tab:variables} describes the variables maintained by each agent for the execution of the algorithm. 
\begin{table}[]
    \centering
    \begin{tabular}{|c|l|}
    \hline
    \textbf{Variable(Phase) } & \textbf{Description}\\
    \hline
        $a.settled(1)$ & Initialised to $\bot$. If the agent $a$ settles at a node then it\\& sets this variable to $\top$   \\
        \hline
        $a.next(1)$ & Initialised to $\bot$. During probing, if an agent $r$ finds an\\& empty node after moving through port, say $p_i$, then the\\& settled agent $a$ sets $a.next=p_i$   \\
        \hline
        $a.prnt\_prt(1)$ & Initialised to $-1$. This variable stores the parent port\\& for the settled agent $a$ \\
        \hline
        $a.vprnt(1, 2)$ & Initised to $0$. If an agent $a$ is acting as virtually settled\\& for its parent $r$ that is out for the execution of phase $2$,\\& then $a$ sets $a.vprnt=1$    \\
        \hline
        $a.v\_prnt\_port(1, 2)$ & Initialised to $-1$. When an agent $a$ settles, it sets\\& $a.v\_prnt\_port=a.parent$. This is to ensure that even if\\& in further rounds it gets included in some other leader's\\& territory, it does not lose knowledge of its parent port\\& from the first DFS traversal \\
        
        \hline
        $a.sticky(1, 2)$  &  Initialised to $0$. In the iteration $1$, the agent $a$ that\\& settles at the root node sets $a.sticky=1$, else for\\& others it is $0$     \\
        \hline
        $a.prm\_leader(1, 2)$ & Initialised to $-1$. If the current round is round $1^1$, the\\& agent $a$ sets $a.prm\_leader=a.leader$. This variable is\\& not updated in any further round   \\
        \hline
        $a.ph(1, 2, 3)$ & Initialised to $0$. Whenever a settled agent $a$ enters into\\& phase $3$, it updates $r_i.ph=3$  \\
        \hline
        $a.prt\_ent(1, 2, 3)$ & Initialised to $-1$. The port used by an agent $a$ to enter\\& into a node is stored in this variable   \\
        
        \hline
        $a.nxt\_prt(2)$ & Initialised to $-1$. This variable stores the port that is\\& to be taken by the $leading$ agent $a$ to traverse according\\& to its DFS traversal   \\
        \hline
        $a.prm\_parent(2)$ & The parent port for the settled agent $a$ maintained\\& according to the DFS traversal in phase $2$ is stored\\& in this variable  \\
        
        \hline
        $a.crnt\_port(3)$ & If $r_{st}$ is present at a node $u$ with a settled agent, say $a$,\\& and $r_{st}.ph=3$, then $a$ sets $a.crnt\_port$ to the smallest\\& child port that is not yet taken by the agent $r_{st}$ to move\\& through in the $explore$ state  \\
        \hline
        $a.final\_set(3)$ & The agents that finally settle at a node after the\\& execution of phase $3$ of the algorithm, set their\\& $final\_set=1$, else it is set to $0$  \\
        \hline
        $a.final\_port(3)$ & The smallest port $i$ of a settled agent $a$ such that the\\& agent settled at the node via port $i$, say $r_j$ has\\& $r_j.final\_set=1$ \\
        \hline
        $a.count(3)$ & Let $a$ has $a.decision=1$. The variable $count$ represents\\& the number of settled agents in the neighborhood of $a$\\& which have $final\_set=0$ i.e., their decision is yet to\\& be taken\\
        \hline
    \end{tabular}
    \\
    \vspace{0.5cm}
    \caption{Description of variables used in different phases maintained by agent $a$ is provided in the table}
    \label{tab:variables}
\end{table}

Note that to make phase $1$ and phase $2$ run in parallel, we make some modifications in the algorithm of \cite{Sudo}, for example, adding a few rounds where agents wait. In the description of each round, we explicitly reference the dispersion technique from \cite{Sudo} along with our adjustments, if and when required, to run phase 1 and phase 2 in parallel. 
Now we are ready to proceed with the details of each round of an iteration.

\vspace{0.2 cm}
Our algorithm proceeds in iterations, where each iteration consists of $24$ rounds. In particular, if the current round is round $t$ and $t=24\cdot i+r$, then $t$ is the $r^{th}$ round of $(i+1)^{th}$ iteration. We denote a round $r$ of an iteration $i$ as $r^i$. This serves two purposes: it minimizes interference between multiple DFSs and enables the simultaneous execution of different phases of the algorithm. Below we describe the functioning of rounds $1-19$ of an iteration $i$. These rounds are particularly for the execution of phases $1$ and $2$. Phase 3 runs from round 20 to round 24 and helps in achieving a D-2-D configuration.
\begin{itemize}
    \item \underline{\textbf{Round $1^i$ \& Round $2^i$}}: (Technique from \cite{Sudo}) Initially, all the agents are leaders, i.e., each agent $r_u$ has $r_u.leader=r_u.ID$. Now, after reaching a new node $v$, there are two possible cases.
    \\
    \begin{itemize}
        \item[(a)] All the agents present at the current node are leaders: In this case, the strongest leader becomes the new leader, while the others become zombies and follow the new leader.  
        \item[(b)] There are leaders and a settler present at the current node: In this case, if the settler is stronger than all the leaders present at the current node, then the leaders become zombies and chase for the leader of the settler by following the $next$ pointer of the settlers at each node. This movement by the zombies occurs in the different dedicated round that happens in rounds $10^i-12^i $.
        However, if one of the leaders, say $r_u$, is the strongest of the remaining leaders and the settler, then the other leaders become zombies and set their leader as $r_u$.   
    \end{itemize}
\vspace{0.2cm}   
    
\noindent{\textbf{Our adjustments}:} Since we run phases $1$ and $2$ in parallel, it may be the case that a node $u$ where an agent $r$ is settled in phase 1 and later at some round, it left $u$ for its conduct of phase $2$ of the algorithm. Note that this can occur only if $r$ is the leading agent i.e. it settled at node $u$ in iteration $i$. In this case, if another agent $v$ reaches $u$ during the rounds when $r$ is not there, $v$ might think that no agent is settled at $u$. To resolve this problem, we divide the task into two rounds.  
\\
        \begin{itemize}
            \item Round $1^i$: If there is a settled agent present at the current node, then the decision is made based on the cases (a) and (b) described above. Otherwise, if there is no settled agent present at the current node, then the decision is made in the round $2^i$. This is because the agents need to verify whether the current node is really vacant or not.

            If $r_v$ is a settled agent with $r_v.vprnt=1$ and it moved in the previous round i.e. in the round $24$ of iteration $i-1$, then $r_v$ returns to its original position via $r_v.prt\_ent$. This movement of $r_v$ is a to-and-fro movement between the node it is originally settled at and the node where a leading agent was settled. This helps the agents to understand whether the node is originally vacant or not. Figure \ref{fig:vprnt} explains the requirement of such an agent.
\\
            \item Round $2^i$: The decision is made based on the cases (a) and (b) described above (if not already taken in the previous round). If $r_v$ is a settled agent with $r_v.vprnt=1$, then it moves through $r_v.v\_prnt\_port$ in this round. 
\end{itemize}
\vspace{0.2cm}
However, if $i=1$, i.e. the current iteration is the first one, then the non-leader agents $r_u$ set $r_u.leader=r_u.prm\_leader=a.ID$, where $a$ is the leader at the current node. The value of $leader$ may change further during phase $1$ (according to the algorithm of \cite{Sudo}).

\begin{figure*}[t!]
    \centering
        \includegraphics[width=.5\linewidth]{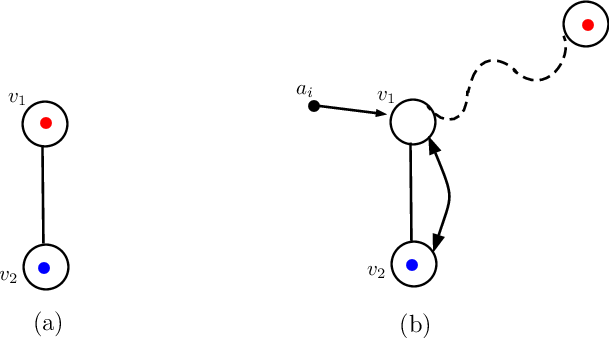}   
    \caption{(a) The red agent is the leading agent that settles at $v_1$. The blue agent settles at $v_2$ and sets its $vprnt=1$. This means it is going to make a to-and-fro movement between $v_1$ and $v_2$ as and when directed (b) The leading agent (red agent) is out from the node $v_1$ for its conduct of phase $2$. An agent $a_i$ visits node $v_1$ and finds it to be vacant. Since blue agent makes hops between $v_1$ and $v_2$, $a_i$ understands that the node $v_1$ is occupied by some agent.}\label{fig:vprnt}
\end{figure*}

\vspace{0.2cm}
\item \underline{\textbf{Round $3^i$ \& Round $4^i$}}: Similar to the cases described for the rounds $1$ and $2$, we need to divide the task into two rounds of $3$ and $4$ to check if there is any settled agent at the current node that is out for its conduct of phase $2$.
    \begin{itemize}
        \item Round $3^i$: In this round, the settler $r_u$ with $r_u.vprnt=1$ that moved in the previous round i.e. round $2^i$ of iteration $i$ returns to its original position via $r_u.prt\_ent$. 
        
        (Technique from \cite{Sudo}) If the leader $a$ finds a settler at the current node, the changes are made based on the following cases, (i) $a$ increments its level by $1$ if it finds a zombie with the same level, (ii) if the settler at this node $w$ is weaker than $a$, then this settler is incorporated into the $a's$ territory by changing its group identifier to $(a.ID, a.lvl)$.
        However, if the leader $a$ finds no settler at the current node, then no changes are made in this round.
        \\
        \item Round $4^i$: In this round if the leader $a$ finds a settler at the current node, then it proceeds similarly to the round $3^i$. However, if the leader $a$ finds no settler at the current node, then this means the current node is unsettled. Thus, the minimum ID zombie settles at this node. However, if there are no accompanying zombies, then $a$ becomes a waiting leader and does nothing until round $1$ of the
        next iteration.
    \end{itemize}

\vspace{0.2cm}
    If $i=1$, then the agent that settles at the vacant node, say $r_p$, sets $r_p.sticky=1$. This agent is responsible for executing phase $2$ of the algorithm. Whereas the leader at the current node, say $a$, sets $a.flag=true$ which implies that the next agent that settles from the group has to virtually take the position of its parent agent whenever required. If the agent $a$ is alone at the current node, then it sets $a.flag_1=true$ besides $a.sticky=1$.
    
    However, if $i>1$ i.e., this is not the first iteration, and there is an agent present at the current node with $flag=true$, then the agent that settles at the current node, say $r_v$, sets $r_v.vprnt=1$. Also, it sets $r_v.v\_prnt\_port = r_v.prnt\_port$ and moves through this port whenever required. Note that the agent $r_v$ stores its value of $prnt\_port$ also in $v\_prnt\_port$ when it settles at a node. This is to ensure that if the agent $r_v$ is included in some other leader's territory then it may change its $prnt\_port$. But the value of $v\_prnt\_port$ does not change. Thus, it continues its to-and-fro movement between its current node and the node where $leading$ agent was settled via the $v\_prnt\_port$.  
\\
    \item \underline{\textbf{Round $5^i$}}: (Technique from \cite{Sudo}) The settled agents $r_u$ with $r_u.help\neq \bot$, moves to its neighbor via the port $r_u.help$ in this round. These agents help in the probing technique for a leader for the iteration $i$.
\\
    \item \underline{\textbf{Rounds $6^i-10^i$}}: (Technique from \cite{Sudo}) Each leader $a$ maintains a flag variable $a.InitProbe\in\{0, 1\}$. This variable is initially set to $1$. This flag variable is raised when the leader $a$ requires probing. Probing is required each time a forward or backward move is made by the leader or level is incremented by it. The subroutine of Probe() is run in rounds $6^i-10^i$. The control of Probe() is returned after the end of round $10^i$ even if the probing is not yet complete. This is completed in the next $6-10$ rounds of the iteration $i+1$.
\\
    \item \underline{\textbf{Rounds $11^i-13^i$}}: (Technique from \cite{Sudo}) In the next three rounds, zombies chase for leaders. A zombie updates its location and swarm levels in rounds $11^i$ and $12^i$. We keep two rounds for this separately, as the situation may arise that the zombies end up at a node that seems unsettled, but the agent settled at it is a $leading$ agent and is out for its execution of phase $2$. The zombie follows the $next$ pointers of the settled agent at the current node. \\
    At the end of round $12^i$, the zombie chases for the leader by moving through the information of the $next$ variable of the settled agent at the current node. The weak zombies move in both rounds $12^i$ and $13^i$, while the strong zombies move only in round $13^i$.
    
\vspace{0.2cm}
    \noindent{\textbf{Our adjustments}:} The settler $c$ with $c.vprnt=1$ moves through $c.v\_prnt\\\_prt$ in round $11^i$ and returns to its original position in round $12^i$. The required information is thus exchanged for the virtually settled node in the round $11^i$ and for the settled agent, which acts as virtually settled for its parent in the round $12^i$.
\\
    \item \underline{\textbf{Rounds $14^i-15^i$}}: (Technique from \cite{Sudo}) If a leader is at a node $w$ and it observes $done=true$ for the settled agent at the current node, then this means the probing for the node $w$ is completed. Based on the probing technique, the leader along with the zombies either makes a $forward$ or a $backward$ move in the round $14^i$. Each time $a$ makes a $forward$ move to a node $u$, it sets $a.prnt\_prt=a.prt\_entered$ after the move. This event occurs in round $15^i$ and only when the last move is $forward$.
\\

    \item \underline{\textbf{Rounds $16^i-19^i$}}: 
    \textbf{(Our Merging Technique)} These rounds are dedicated to the movement of the $leading$ agent that executes phase $2$ of the algorithm. In this phase, such an agent moves through the nodes of the graph where an agent is settled and constructs its new DFS tree. Based on the $prm\_leader$ value of the settled agent, either the $leading$ agent expands its newly constructed DFS tree or waits at that node itself. The aim is to construct a single DFS tree by the end of this phase. Let a $leading$ agent be denoted by $r_{st}$ and is currently present at a node $v$. Let $r_{st}.flag_1=false$. It moves through the port $r_{st}.nxt\_prt$ in the round $16^i$. Suppose $r_{st}$ enters the node $v_1$. In the round $17^i$, it checks if there is an agent settled at the node $v_1$. 
    If yes, then $r_{st}$ stays at the node $v_1$ and the agent settled at the node $v_1$, say $r_1$, makes amendments in its variables for phase $2$ as described further.
Otherwise, the agent $r_{st}$ waits for the next round, i.e., round $18^i$. If there is a settled agent at node $v_1$, then changes are made based on the cases described further. However, if neither in round $17^i$ nor $18^i$, a settled agent is present at the node $v_1$, then the agent $r_{st}$ decides to return to its previous node $v$ in round $19^i$ and wait till the next available $16^{th}$ round. This is because the agent $r_{st}$ understands that since there is no settled agent at the current node $v_1$ means dispersion is not yet completed. Thus it returns to its previous node and re-checks the node $v_1$ in the next round $16^i$. In other words, the agent $r_{st}$ continues its movement from one node to another according to its DFS traversal only if it finds a settled agent at that node. Otherwise, it waits for that vacant node to be settled so that it can proceed further.  
The decision is made by the settled agent $r_1$ at the node $v_1$ and the agent $r_{st}$ based on the following cases if the agent $r_{st}$ has $state=explore$:
\\
    \begin{itemize}
        \item If the settled agent $r_1$ has $r_1.prm\_leader<r_{st}.prm\_leader$, then the agent $r_1$ updates $r_1.met\_prm=r_1.prm\_leader=r_{st}.prm\_leader$, and $r_1.prm\_parent=r_{st}.prt\_ent$. In other words, $r_{st}$ met with $r_1$ for the first time during the construction of its DFS tree and $r_1.prm\_leader<r_{st}.prm\_leader$. Thus, $r_1$ updates its $prm\_parent$ to the $prt\_ent$ value of $r_{st}$. Note that $r_1$ also updates its $met\_prm$ value to ensure that further if $r_{st}$ visits $r_1$ in the $explore$ state, then $r_{st}$ understands that it is already visited and can $backtrack$ from that node to construct its DFS tree in the correct manner. Finally, $r_{st}$ updates $r_{st}.nxt\_prt=(r_{st}.prt\_ent+1)\mod\delta(u)$. If $r_{st}.nxt\_prt=r_{st}.prt\_ent$, then it updates $r_{st}.state=backtrack$. It moves through this updated $r_{st}.nxt\_prt$ in the next $16^{th}$ round. 
        \item If the settled agent $r_1$ has $r_1.prm\_leader>r_{st}.prm\_leader$, then $r_{st}$ stops the further execution of the algorithm and waits at its previous node. This is to ensure that the $leading$ agent with the largest value of $prm\_leader$ completes the construction of its DFS tree.
        \item If the settled agent $r_1$ has $r_1.prm\_leader=r_{st}.prm\_leader$ but $r_1.met\\\_prm\neq r_{st}.prm\_leader$, then the agent $r_1$ updates $r_1.prm\_parent=r_{st}.prt\_ent$. This means $r_1$ belongs to the territory of $r_{st}$ but it is visited for the first time during the construction of DFS tree by $r_{st}$. $r_{st}$ further updates $r_{st}.nxt\_prt=(r_{st}.prt\_ent+1)\mod\delta(u)$. If $r_{st}.nxt\_prt=r_{st}.prt\_ent$, then it updates $r_{st}.state=backtrack$. It moves through this updated $r_{st}.nxt\_prt$ in the next $16^{th}$ round. 
        \item If the settled agent $r_1$ has $r_1.met\_prm=r_{st}.prm\_leader$, then the agent $r_{st}$ sets its state to $backtrack$ and returns to the previous node via $r_{st}.prt\_ent$ in the next $16^{th}$ round. 
    \end{itemize}
\vspace{0.2cm}
    If the agent $r_{st}$ has $state=backtrack$, then the settled agent present at the current node $u$, say $r_1$, definitely has $r_1.prm\_leader=r_{st}.prm\_leader$ and $r_1.met\_prm=r_{st}.prm\_leader$. Thus, the following cases are checked: 
    The agent $r_{st}$ increments $r_{st}.prt\_ent=(r_{st}.prt\_ent+1)\bmod \delta(u)$.
\\
    \begin{itemize}
        \item If the updated value of $r_{st}.prt\_ent\neq r_{1}.prm\_parent$, then $r_{st}.nxt\_prt$ is set to $r_{st}.prt\_ent$ and $r_{st}.state=explore$. It then moves to the adjacent node in the next $16^{th}$ round.
        \item   However, if the updated value of $r_{st}.prt\_ent= r_{1}.prm\_parent$, then $r_{st}.nxt\_prt$ is set to $r_{1}.prm\_parent$ and $r_{st}.state=backtrack$. It then moves to the previous node in the next $16^{th}$ round.
    \end{itemize}

Note that if $r_{st}$ is a leading agent and $r_{st}.flag_1=true$, then it was alone in the initial configuration\footnote{In the initial configuration, if an agent is present alone at a node, this case is addressed differently.}. Hence, there is no agent that can take its position virtually when it is out for its execution of phase $2$. When it moves through its $nxt\_prt$ in round $16^i$, if it finds an agent $r_u$ settled that has $r_u.prm\_leader>r_{st}.prm\_leader$, then $r_{st}$ moves to its original node where it was settled, sets $r_{st}.flag_1=false$, and stops its further execution of phase $2$. However, if $r_u.prm\_leader<r_{st}.prm\_leader$, then $r_u$ updates $r_u.vprnt=1$ and acts as virtually settled for the node where $r_{st}$ was originally settled. Furthermore, the agent $r_{st}$ proceeds based on the cases discussed above. With this, we handle the case when there is only a single agent present at a node in the initial configuration.\\  
    On the other hand, the settled agent $c$ with $c.vprnt=1$ moves through $c.v\_prnt\_port$ in round $16^i$ and returns to its original position in round $17^i$.
    Note that each time the agent $r_{st}$ visits a settled agent and makes it a part of its DFS, the settled agent also maintains the distance to be either $0$ or $1$ if it is at an even or an odd distance, respectively, w.r.t. the DFS traversal by the agent $r_{st}$ from its root. 
    After a $leading$ agent, say $r_{l_{min}}$ traverses through the graph and returns to its root node, it understands that it has completed phase $2$ of the algorithm. After this completion, the agent $r_{l_{min}}$ has all the remaining $leading$ agents that stopped their execution of phase $2$ due to their meeting with a settled agent with a larger value of $prm\_leader$. Thus, now the agent $r_{l_{min}}$ re-traverses through the graph. It moves from a node $u$ to a node $v$ in round $16^i$ following its DFS traversal. While the agent $r$ with $r.vprnt=1$ continues its movement as described earlier. 
    In the round $17^i$, if $r_{l_{min}}$ meets with an agent $r$ at $v$ with $r.vprnt=1$, then this means the current node $v$ is virtually settled by the agent $r$. Thus, the minimum ID $leading$ agent, say $r_p$, from the group of $leading$ agents present with $r_{l_{min}}$ settles at the node $v$ and sets $r_p.prm\_parent=r_{l_{min}}.prt\_ent$. However, the agent $r$ sets $r.vprnt=0$ and moves to its original position in this round. An example illustrating the execution of phase $2$ of the algorithm is provided in the figure \ref{fig:phase 2}.

    \begin{algorithm}\small
    \caption{The behavior of a leader $a$} \label{alg:leader_1_4}
    {
    \tcc{***** \textbf{Round $1$} *****}
    \If{there is a settled agent present at the current node}
    {
        Let $a$ be currently placed at node $w$\\
        \If{there is an agent $b$ that is either a settler or a leader and $b \succ a$  }
        {
            $a.lvl_L\leftarrow a.lvl_S \leftarrow a.lvl$\\
            $a.leader\leftarrow b.leader$
        }
    }
    \ElseIf{there is no settled agent present at the current node}
    {
        do nothing
    }

    \tcc{***** \textbf{Round $2$} *****}
    \If{there is a settled agent at the current node}
    {
        \If{no changes were made in the previous round}
        {
            \If{there is an agent $b$ that is either a settler or a leader and $b \succ a$  }
            {
            $a.lvl_L\leftarrow a.lvl_S \leftarrow a.lvl$\\
            $a.leader\leftarrow b.leader$
            }
        }
    }
    
        \tcc{***** \textbf{Round $3$} *****}
    \If{there is a settled agent $r_w$ at the current node}
        {
        \If{$r_w\succ a$}
        {
           $r_w.parent \leftarrow a.prt\_ent$\\
           $(r_w.leader,r_w.lvl) \leftarrow (a.ID,a.lvl)$
        }
        \ElseIf{there is a zombie $b$ such that $b.lvl=a.lvl$}
            {
            $(a.lvl, b.lvl)\leftarrow (a.lvl+1, 0)$\\
            $r_w.parent=\bot$ and$a.InitProbe=1$\\
            $(r_w.leader, r_w.lvl)\leftarrow (a.ID, a.lvl)$\\
            \If{$a.InitProbe=1$}
            {
                $(r_w.next, r_w.checked, r_w.help, r_w.done)\leftarrow (\bot, -1, \bot, false)$\\
                update $a.InitProbe=0$\\
                call Probe(a)
            }
            }
        }
    \tcc{***** \textbf{Round $4$} *****}
    \If{no changes were made in the previous round}
    {
        \If{there is no settled agent at the current node}
        {
            \If{$a$ is an active leader i.e. there is at least one zombie}
            {
                the minimum ID zombie $r_w$ settles at $w$\\
                $r_w.parent=a.prt\_ent$\\
                $(r_w.leader, r_w.lvl)\leftarrow (a.ID, a.lvl)$\\
                \If{$a.flag=false$}
                {
                    \If{the current iteration is the first iteration}
                    {
                    set $r_w.sticky=1$\\
                    leader $a$ updates $a.flag=true$
                    }
                }
                \Else
                {
                    $r_w$ sets $r_w.vprnt=1$ and $r_w.v\_prnt\_port=r_w.prnt\_port$\\
                    $a$ updates $a.flag=false$
                }
                
                call Probe(a)
            }
        }
        \Else
        {
            \If{$r_w\succ a$}
            {
           $r_w.parent \leftarrow a.prt\_ent$\\
           $(r_w.leader,r_w.lvl) \leftarrow (a.ID,a.lvl)$
            }
        \ElseIf{there is a zombie $b$ such that $b.lvl=a.lvl$}
            {
            $(a.lvl, b.lvl)\leftarrow (a.lvl+1, 0)$\\
            $r_w.parent=\bot$ and $a.InitProbe=1$\\
            $(r_w.leader, r_w.lvl)\leftarrow (a.ID, a.lvl)$\\
            \If{$a.InitProbe=1$}
            {
                $(r_w.next, r_w.checked, r_w.help, r_w.done)\leftarrow (\bot, -1, \bot, false)$\\
                update $a.InitProbe=0$\\
                call Probe(a)
            }
            }   
        }
    }

}
\end{algorithm}

\begin{algorithm}\small
    \caption{The behavior of a settled agent $r_u$ in round $5$}\label{alg:settled_2}
    \tcc{***** \textbf{Round $5$} *****}
    \If{$r_u.help\neq \bot$}
    {
        move through the port $r_u.help$
    }
\end{algorithm}

\begin{algorithm}\small
    \caption{Probe(a)}\label{alg:probing}
    \tcc{***** \textbf{Round $6$} *****}
    Let $w$ be the current node where $a$ is present and $r_w$ be the settled agent at $w$\\
    
 $r_w.next\leftarrow \begin{cases}
       \min P & if ~P\neq \emptyset\\ 
       \bot & otherwise
   \end{cases}$
   
    where $P=[0, r_w.checked]\setminus\{b.prt\_ent | b \text{ is a settled agent at $w$ except $r_w$}\}$ \\
    $b.help\leftarrow\bot$ for all settled agents $b$ such that $b\prec a$\\
    All the agents $b$ go back to their original positions\\
    
    \tcc{***** \textbf{Round $7$} *****}
    \If{$r_w.next \neq \bot \lor r_w.checked = \delta(w)-1$}
    {
        update $b.help\leftarrow \bot$ for all settled agents $b$ at $w$ except $r_w$\\
        All such agents $b$ move to their original positions\\
        update $r_w.done=true$
    }
    \tcc{***** \textbf{Round $8$} *****}
    \Else
    {
       Let $\{a_1, a_2,..., a_x\}$ be the set of agents at the node $w$ except $r_w$\\
       Let $\delta'=\min(x, \delta(w)-1-r_w.checked)$\\
       Let $u_i=N(w, i+r_w.checked)$ for $i=1, 2, ..., \delta'$\\
       \For{$a_i \in \{a_1, a_2,..., a_x\} $}
       {
       Each agent $a_i$ in parallel moves to the node $u_i$

    \tcc{***** \textbf{Round $9$} *****}
        \If{$(a_i.leader, a_i.lvl)=(r_{u_i}.leader, r_{u_i}.lvl)$}
        {
            set $a_i.found=true$\\
            $r_{u_i}.help=a_i.prt\_ent$
        }
        \Else
        {
            set $a_i.found=false$
        }
        move through $a_i.prt\_ent$
        }

    \tcc{***** \textbf{Round $10$} *****}
    \If{$\exists i \in [1, \delta']: a_i.found=false$}
    {
        $r_w.next\leftarrow i+r_w.checked$
    }
    $r_w.checked=r_w.checked+\delta'$\\
    All the settled agents at $w$ return to their original positions
}
    
\end{algorithm}

\begin{algorithm}\small
    \caption{The behavior of a zombie $z$ in rounds $11-13$}\label{alg:zombie_11_13}
    \tcc{***** \textbf{Round $11$} *****}
    $(z.lvl_L, z.lvl_S)\leftarrow (\psi(w).lvl, \max\{z'.lvl | z' \text{is a zombie at the current node}\})$\\
    \If{there is no leader at the current node and $z$ is a weak zombie}
    {
        Move to $N(\nu(z), \psi(\nu(z)).next)$
    }
    \tcc{***** \textbf{Round $12$} *****}
    $(z.lvl_L, z.lvl_S)\leftarrow (\psi(w).lvl, \max\{z'.lvl | z' \text{is a zombie at the current node}\})$\\
    \If{there is no leader at the current node and $z$ is a weak zombie}
    {
        Move to $N(\nu(z), \psi(\nu(z)).next)$
    }
    \tcc{***** \textbf{Round $13$} *****}
    \If{there is no leader at the current node}
    {
        Move to $N(\nu(z), \psi(\nu(z)).next)$
    }
    
\end{algorithm}

\begin{algorithm}\small

    \caption{Algorithm for agent $r_{st}$ with $r_{st}.sticky=1$} 
    \label{alg:r_st_checking}
    \tcc{***** \textbf{Round $16$} *****}
    move through $r_{st}.nxt\_prt$\\
    \tcc{***** \textbf{Round $17$} *****}
    \If{there is a settled agent $r_i$ at the current node}
    {
    call function Check($r_{st}$, $r_i$)
    }
    \ElseIf{there is no settled agent at the current node}
    {
    do nothing
    }
    \tcc{***** \textbf{Round $18$} *****}
    \If{no changes were made in round $17$}
    {
        \If{there is a settled agent $r_i$ at the current node}
        {
        call function Check($r_{st}$, $r_i$)
        }
        \ElseIf{there is no settled agent}
        {
    \tcc{*****\textbf{Round $19$} *****}
        return to the previous node via $r_{st}.prt\_ent$
        }
    }
\end{algorithm}

\begin{algorithm}\small
\tcc{comparing the value of $prm\_leader$ of the settled agent}
    \caption{function Check($r_{st}$, $r_i$)}\label{alg:check}
    \If{$r_{st}.state=explore$}
{
    \If{the settled agent $r_i$ has $r_i.prm\_leader<r_{st}.prm\_leader$}
    {
        $r_i.met\_prm=r_i.prm\_leader=r_{st}.prm\_leader$\\
        $r_i.prm\_parent=r_{st}.prt\_ent$
    }
    \ElseIf{the settled agent $r_i$ has $r_i.prm\_leader>r_{st}.prm\_leader$}
    {
    the agent $r_{st}$ moves via $r_{st}.prt\_ent$ and waits at the node
    }
    \ElseIf{the settled agent $r_i$ has $r_i.prm\_leader=r_{st}.prm\_leader$ but $r_i.met\_prm\neq r_{st}.prm\_leader$}
    {
    agent $r_i$ updates $r_i.prm\_parent=r_{st}.prt\_ent$
    }
    \ElseIf{the settled agent $r_i$ has $r_i.met\_prm=r_{st}.prm\_leader$}
    {
    set $r_{st}.state=backtrack$ and $r_{st}.nxt\_prt=r_{st}.prt\_ent$
    }
}
\ElseIf{$r_{st}.state=backtrack$}
{
update $r_{st}.prt\_ent=(r_{st}.prt\_ent+1)\mod{\delta(u)}$\\
\If{$r_{st}.prt\_ent\neq r_u.prm\_parent$}
{
set $r_{st}.nxt\_prt=r_{st}.prt\_ent$\\
$r_{st}.state=explore$
}
\Else
{
set $r_{st}.nxt\_prt=r_{u}.prm\_parent$\\
}
}

\end{algorithm}

\begin{algorithm}\small
    \caption{Algorithm for settled agent $r_j$ from rounds $20-23$}\label{alg:settled_20_23}
    \tcc{***** \textbf{Round $20$} *****}
    \If{$r_j$ is a settled agent and $r_j.dist=0$}
    {
    move through $r_j.prm\_parent$
    }
    \tcc{***** \textbf{Round $21$ }*****}
    \If{$r_j$ is a settled agent with $r_j.dist=1$ and agent $r_{st}$ with $r_{st}.ph=3$ is present with $r_j$}
    {
    $r_j$ updates $r_j.crnt\_port$ using the incoming ports of the agents that entered the current node in round $20$
    }
    \If{$r_j$ moved in round $20$}
    {
    move through $r_j.prt\_ent$
    }
    \tcc{***** \textbf{Round $22$} *****}
    \If{$r_j$ is a settled agent and $r_j.dist=1$}
    {
    move through $r_j.prm\_parent$
    }
    \tcc{***** \textbf{Round $23$} *****}
    \If{$r_j$ is a settled agent with $r_j.dist=0$ and agent $r_{st}$ with $r_{st}.ph=3$ is present with $r_j$}
    {
    $r_j$ updates $r_j.crnt\_port$ using the incoming ports of the agents that entered the current node in round $22$
    }
    \If{$r_j$ moved in round $22$}
    {
    move through $r_j.prt\_ent$
    }
\end{algorithm}

\vspace{0.2cm}
    \item \underline{\textbf{Rounds $20^i-21^i$}}: In the round $20^i$, each agent $r_j$ that is settled at an even distance from the root node (w.r.t. distance in the DFS traversal by $r_{st}$ in the phase $2$), moves through $r_j.prm\_parent$. 

    However, in round $21^i$, the settled agent at odd distant node $u$ updates its $crnt\_port$ if $r_{st}$ is present with it and has $r_{st}.ph=3$, by seeing the agents who moved in the $20^{th}$ round. The agents $r_j$ that moved in the $20^{th}$ round return to their original position via the port they entered through.
\\
    \item \underline{\textbf{Rounds $22^i-23^i$}}: In the round $22^i$, each agent $r_j$ that is settled at an odd distance from the root node (w.r.t. distance in the DFS traversal by $r_{st}$ in the phase $2$), moves through $r_j.prm\_parent$. 

    However, in round $23^i$, the settled agent at even distant node $u$ updates its $crnt\_port$ if $r_{st}$ is present with it and has $r_{st}.ph=3$, by seeing the agents who moved in the $22^{nd}$ round. The agents $r_j$ that moved in the $22^{nd}$ round return to their original position via the port they entered through.
\\
    \item \underline{\textbf{Round $24^i$}}: The agents run phase $3$ in this round. The phase $3$ of the algorithm runs as follows. 
    The $leading$ agent that returns to its original position i.e., the root, and has no further ports left to be explored, understands that phase $2$ has been completed. This implies that a single DFS corresponding to the $prm\_leader$ of this $leading$ agent exists in the graph. The phase $3$ of the algorithm is then initiated by this agent. It updates $r_{st}.ph=3$, and $r_{st}.state=explore$.
    The agent $r_{st}$, from node $u$, moves through the $r_u.crnt\_port$ information, where $r_u$ is the settled agent at $u$, to reach node $r_k$. If $r_k.crnt\_port\neq \bot$, then $r_{st}$ further moves through $r_k.crnt\_port$. It moves through the $crnt\_port$ value of each settled agent unless it reaches a settled agent, say $r_l$, that has $r_l.crnt\_port=\bot$, i.e., $r_l$ has no further child ports. Thus, the agent $r_l$ sets $r_l.decision=1$ and $r_{st}$ waits at the node for $2\delta(v)$ rounds, where $v$ is the current node.
    The agent $r_l$ moves through its ports one by one and does the following. 
    \\

    \begin{itemize}
    \item If there is at least one neighbor, say $r_j$, of $r_l$ with $r_j.final\_set=1$, then $r_l$ settles at the node where $r_j$ is settled and $r_l.final\_set=1$. 
    \item If there are no neighbors of $r_l$ with $final\_set=1$, then $r_l$ settles at its original position and sets $r_l.final\_set=1$.
    \end{itemize}
\vspace{0.2cm}
    The settled agent on the other hand when meet with an agent $r_l$ with $r_l.decision=1$, it decrements its value of $count$ and finally terminate when its value of $count$ is $0$. After visiting all its ports, if $r_l.final\_port$ is $\bot$ then it settles at its original position and sets $r_l.final\_set=1$. However, if $r_l.final\_port\neq\bot$, the agent $r_l$ moves through this $final\_port$ after visiting through all its ports, and settles there, after setting $r_l.final\_set=1$ and terminates.

    After the wait of $2\delta(v)$ rounds, the agent $r_{st}$ backtracks to the previous node where an agent, say, $r_n$ is settled, and checks the value of $r_n.crnt\_port$. If $r_n.crnt\_port\neq\bot$ then the agent $r_{st}$ changes its state to $explore$ and moves through the value of $r_n.crnt\_port$. However, if $r_n.crnt\_port=\bot$, then $r_{st}$ waits for $2\delta(r_n)$ rounds while $r_n$ sets $r_n.decision=1$.      
\end{itemize}

In this way, the three phases of the algorithm are run such that dispersion is achieved followed by the merging of all the DFSs into a single DFS tree and finally conversion of the dispersed configuration into a D-2-D one. Thus, D-2-D with termination is achieved by the agents without any prior knowledge of any of the global parameters. An example illustrating the execution of phase $3$ of the algorithm is presented in the figure \ref{fig:phase 3}.

The behavior of a leader in rounds $1-4$ is given in the Algorithm \ref{alg:leader_1_4}. The behavior of a settled agent in round $5$ is given in the Algorithm \ref{alg:settled_2}. The pseudo-code of Probing that proceeds in rounds $6-10$ is given in the Algorithm \ref{alg:probing}. A zombie functions in rounds $11-13$ and its pseudo-code is given in Algorithm \ref{alg:zombie_11_13}. The pseudo-code for agent with $sticky=1$ in rounds $16-19$ is given in Algorithm \ref{alg:r_st_checking}. Function Check() in which $leading$ agent compares its $prm\_leader$ value with that of the settled agent is given in Algorithm \ref{alg:check}. The algorithm for a settled agent in rounds $20-23$ is given in Algorithm \ref{alg:settled_20_23}. The algorithm for an agent with $ph=3$ in round $24$ is given in Algorithm \ref{alg: movement_r_st}.

\begin{figure}[ht!]
  \centering
  
  \includegraphics[width=.25\linewidth]{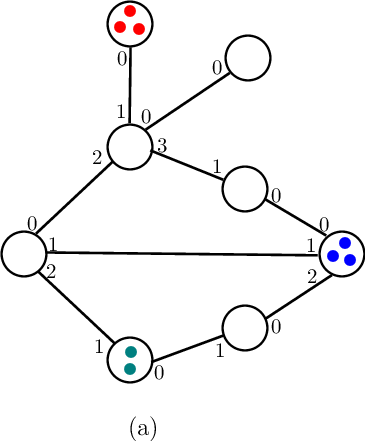} \hspace{0.1cm} 
  \includegraphics[width=.25\linewidth]{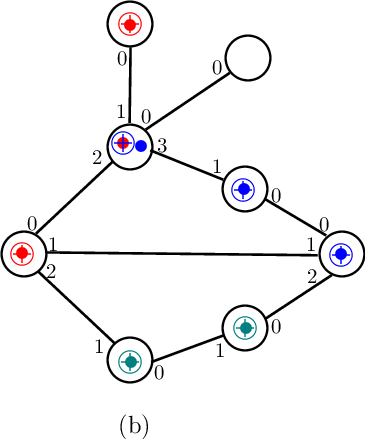}  \hspace{0.1cm}
  \includegraphics[width=.25\linewidth]{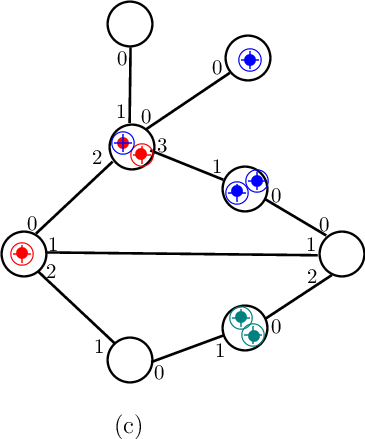}  \hspace{0.1cm}
  \\
  \vspace{0.2cm}
  \includegraphics[width=.25\linewidth]{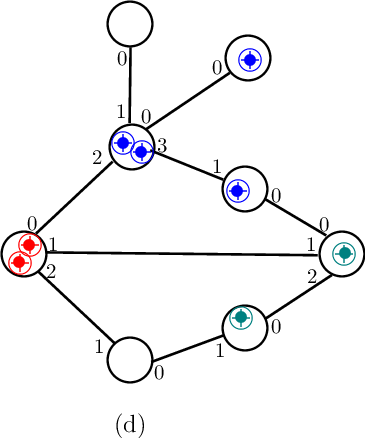}  \hspace{0.1cm}
  \includegraphics[width=.25\linewidth]{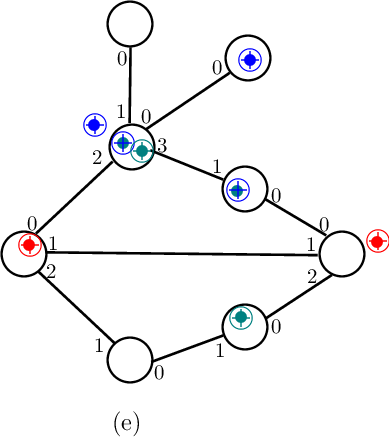} \hspace{0.1cm}
 \includegraphics[width=.25\linewidth]{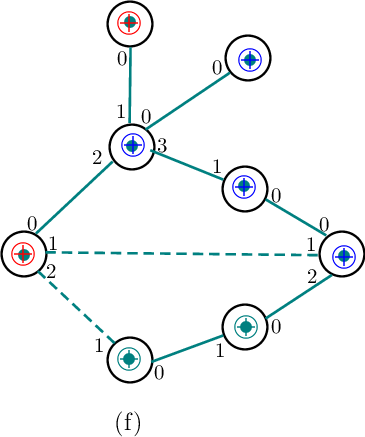}  \hspace{0.1cm}
\caption{(a) The initial configuration of the agents. Let red(`r'), blue(`b'), and green(`g') dots represent three different groups with $prm\_leader$ values as r, g, and b respectively ($r<b<g$). Let the dot represent the $prm\_leader$ and the outer circle represent its group in the dispersion phase, i.e. phase $1$ (b) The group of agents begin dispersion. The outer circle represents the group it is a part of (c) The $leading$ agents from each group start their respective DFS (d) The `b' $leading$ agent meets with a settled agent that has $prm\_leader=r$. Thus, its $prm\_leader$ value is changed to `b' (e) The `r' $leading$ agent waits at a node as it meets with a settled agent with $prm\_leader=b$(this is represented with the agent waiting outside the node). The `g' $leading$ agent meets with the settled agents with $prm\_leader=b$. Thus, these settled agents update their $prm\_leader=g$. While the `b' $leading$ agent waits at the node where it meets with the settled agent having $prm\_leader=g$ (f) Finally all the waiting $leading$ agents are picked by the `g' $leading$ agent and positioned at the virtually settled nodes. All the settled agents now have $prm\_leader=g$ and there is a single DFS tree(darkened lines represent the tree edge and the dotted lines represent the non-tree edges) in the graph corresponding to the DFS by `g' $leading$ agent.         }
\label{fig:phase 2}
\end{figure}
\begin{figure}[ht!]

  \includegraphics[width=.30\linewidth]{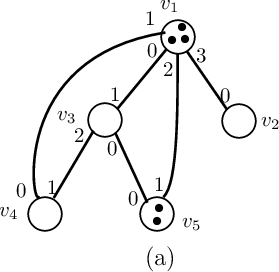}  
  \includegraphics[width=.30\linewidth]{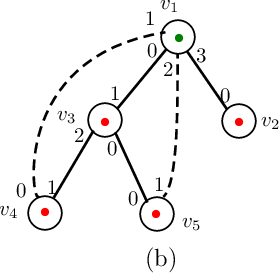}  
   \includegraphics[width=.30\linewidth]{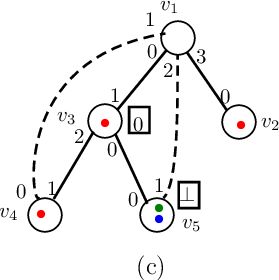}  
   \includegraphics[width=.30\linewidth]{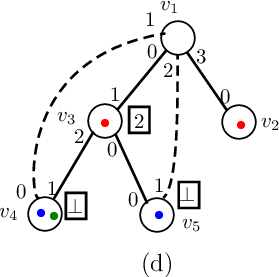}  \hspace{0.3cm}
  \includegraphics[width=.30\linewidth]{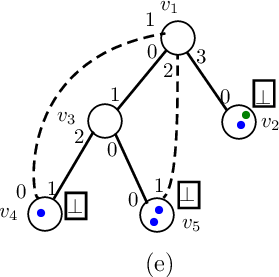}  \hspace{0.3cm}
   \includegraphics[width=.30\linewidth]{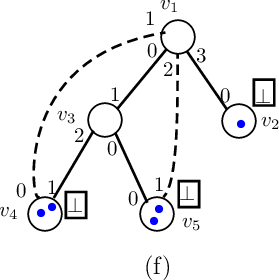}  

\caption{(a) agents are initially positioned arbitrarily in the graph (b) After the completion of phase $1$ and phase $2$ of the algorithm, there exists a single DFS tree with tree edges (shown as darkened lines) and non-tree edges (shown as dotted lines). The DFS tree corresponds to the DFS traversal by the agent $r_{st}$ with the largest value of $prm\_leader$ (denoted with the green colour) (c) $r_{st}$ moves through $crnt\_port$ values of the agents settled at $v_1$ and $v_3$ eventually reaching node $v_5$ which has no child ports. The agent that is settled at $v_5$ sets $decision=1$ and checks if there is any settled agent in its one-hop neighbourhood with $final\_set=1$. Failing to find one, it finally settles at its node by setting $final\_set=1$ which is shown by a blue dot (d) $r_{st}$ backtracks to $v_3$ and further moves through $crnt\_port=2$ to reach $v_4$. Following the decision of the agent settled at $v_4$, $r_{min}$ backtracks to node $v_3$ (where $crnt\_port=\bot$), the decision for this node is now taken. The agent settled at $v_3$ finds an agent at $v_3$ that has $final\_set=1$. Thus, the agent settled at $v_3$ vacates its position and settles at the node $v_5$. (e) $r_{st}$ arrives at $v_1$ which now has $crnt_port=3$ and moves through this port to reach $v_2$. Post the decision for this node, finally $r_{st}$ backtracks to $v_1$ (which now has $crnt\_port=\bot$). The decision for the agent $r_{st}$ is taken. It finds an agent with $final\_set=1$ in its one-hop. Thus, it moves through $final\_port=0$ and sets $r_{st}.final\_set=1$ (f) The D-2-D configuration. }
\label{fig:phase 3}
\end{figure}

\begin{algorithm}[ht!]
    \caption{Algorithm for agent $r_{st}$ with $r_{st}.ph=3$ } \label{alg: movement_r_st}
    /***** \textbf{Round $24$} *****/
    
    \tcc {Let $r_{st}$ be present at node $v$ where agent $r_j$ is settled} 
    
    \If{$r_{st}.state=explore$}
    {
        \If{$r_j.crnt\_port=\bot$}
        {
        set $r_{st}.prt\_ent=r_j.parent$\\
        wait for $48\delta(v)$ rounds \\
        set $r_{st}.state=backtrack$\\
        move through $r_{st}.prt\_ent$
        }
        \Else
        {
        move through $r_j.crnt\_port$\\
        }
    }
    \ElseIf{$r_{st}.state=backtrack$}
    {
        \If{$r_j.crnt\_port \neq \bot$}
        {
        set $r_{st}.state=explore$\\
        move through $r_j.crnt\_port$
        }
        \Else
        {
        set $r_{st}.prt\_ent=r_j.parent$\\
        wait for $48\delta(v)$ rounds\\
        move through $r_{st}.prt\_ent$
        }
    } 
\end{algorithm}

\begin{algorithm}[ht!]
    \caption{Algorithm for a settled agent $r_j$ in phase $3$} \label{alg: settled_ph3}
       
    \tcc{***** \textbf{Round $24$} *****}
    \If{$r_j.final\_set=0$}
    {
    \If{agent $r_{st}$ is present with $r_j$ at node $v$ and $r_j.crnt\_port=\bot$} 
    {
    set $r_j.decision=1$\\
    \For{$j=0$ to $\delta(v)-1$}
    {
        move through the port $j$\\
        \If{the settled agent has $final\_set=0$}
        {
        set $r_j.count=r_j.count+1$\\
        }
        \ElseIf{the settled agent has $final\_set=1$}
        {
            \If{$r_j.final\_port=-1$}
            {
            set $r_j.final\_port=j$
            }
            \Else
            {
            do not update $r_j.final\_port$
            }
        }
        move back to its original position $v$
    
    }
    
    \If{$r_j.final\_port\neq-1$}
    {
    move through $r_j.final\_port$\\
    set $r_j.final\_set=1$ and terminate\\
    }
    \Else
    {
    set $r_j.final\_set=1$
    }
    }
    }
    \If{$r_j.final\_set=1$}
    {
    each time a settled agent $r_i$ with $r_i.decision=1$ visits $r_j$, $r_j$ decrement $r_j.count$ by $1$\\
    terminate $r_j$ when $r_j.count=0$
    }    
\end{algorithm}

\section{Correctness and Analysis}
In this section, we provide the correctness of our algorithm and its analysis. 
The following lemma proves that phases $1$ and $2$ of the algorithm can be correctly run in parallel by the agents.
\begin{lemma} \label{lem:ph2_no_interruption}
   None of the variables maintained by the agents for the execution of phase $1$ (or phase $2$) is altered due to the execution of phase $2$ (or phase $1$) nor does any settled agent change its position to settle at another node.
\end{lemma}
\begin{proof}
    The phase $1$ of the algorithm proceeds in rounds $1-15$. The agent settled at the root node of a group, namely the ``$leading$ agent'' moves from its original position to reach a new node, say $v$, only if there is a settled agent present at node $v$, otherwise, it moves back to its previous node. The $leading$ agent performs such movements in the rounds $16-19$. The primary function of these $leading$ agents is to merge different DFS trees existing in the graph. However, the agents settled at the nodes during phase $1$ do not change their position due to phase $2$. Agents maintain distinct variables for Phases $1$ and $2$, and no variable maintained for Phase $1$ is modified during Phase $2$, and vice versa. The only potential issue in Phase $1$ is that the movement of the leading agent might result in a node being misinterpreted as unsettled during dispersion. This issue is addressed by employing another settled agent $r$ with $r.vprnt=1$ that moves back and forth between its position and the node where the leading agent is settled. In each iteration, specific rounds are designated for these agents to move. Consequently, if a group of unsettled agents reaches such a node and encounters agent $r$ with $r.vprnt=1$, it understands that the current node is not vacant. Thus, phases $1$ and $2$ can be run in parallel without affecting the execution of one to the other.
    
    The agent settled at the root node of a group, namely the ``$leading$ agent'' moves from its original position to reach a new node, say $v$, only if there is a settled agent present at node $v$, otherwise, it moves back to its previous node. The $leading$ agent performs such movements in the rounds $16-19$. Whereas, the phase $1$ proceeds in the rounds $1-15$. The movement performed by the $leading$ agents leaves their originally settled nodes vacant. This may lead to misinterpretation of the node to be unsettled during dispersion (phase $1$) that happens in the rounds $1-15$. To resolve this issue, the settled agents $r$ with $r.vprnt=1$ move to and fro from their positions to their parent nodes (the nodes where these $leading$ agents were settled at) in each iteration. Thus, phases $1$ and $2$ are executed without any interruption between each other.        
\end{proof}

The following lemma proves that the agents correctly achieve dispersion after the completion of phase $1$ of the algorithm.

\begin{lemma}
  After the completion of phase $1$ of the algorithm, at each node $v$ of the graph, there is a settled agent $r$ that either has $r.vprnt=0$ or $r.vprnt=1$. The agent $r$ with $r.vprnt=1$ acts as virtually settled for its parent node.  
\end{lemma}
\begin{proof}
   The lemma \ref{lem:ph2_no_interruption} shows that the phase $2$ of the algorithm does not affect the dispersion of the agents. Thus, the correctness of the phase $1$, i.e., the dispersion phase holds directly from the analysis given in \cite{Sudo}. However, the only issue that may arise is due to the $leading$ agents that leave their settled positions and traverse through the graph for the execution of phase $2$. For such $leading$ agents, settled agents $r$ with $r.vprnt=1$ move to and fro between their parent node and their original node. For the movement of these agents, certain rounds are allotted as mentioned in the section \ref{sec:algorithm}. Thus, during dispersion, if the group of unsettled agents reaches a node, then the agents perform their computations at a node in a specific round that is after the rounds for the movement of the agents with $vprnt=1$. Thus, the group of unsettled agents can never misinterpret a settled node to be vacant. Finally, after the completion of phase $1$, either a settled agent has $r.vprnt=0$ (that is settled at its original node) or a settled agent has $r.vprnt=1$ (that is settled at its node and is working as virtually settled for its parent as well).   
\end{proof}

\begin{lemma}\label{lem:leading_agents}
    By the end of phase $2$ of the algorithm, there is exactly one DFS tree present in the graph.
\end{lemma}
\begin{proof}
    In phase $2$ of our algorithm, the $leading$ agent moves through the settled agents and checks whether the $prm\_leader$ value of the settled agent is larger or smaller than that of the $leading$ agent (refer Algorithm \ref{alg:r_st_checking}). If the settled agent has a larger value of $prm\_leader$, then the $leading$ agent makes no movement further and waits at the current node itself. However, if the value of $prm\_leader$ of the settled agent is smaller than that of $leading$ agent, then the settled agent updates its $prm\_leader$ value to be the same as that of $leading$ agent (refer to lines $2-5$ of Algorithm \ref{alg:check}). In other words, in this case, the $leading$ agent can proceed further to a new node. Thus, the $leading$ agent with the largest value of $prm\_leader$ is never stopped at any node and it completes the DFS traversal of the whole graph. Thus, each settled agent has $prm\_leader$ value same as that of the $leading$ agent with the largest value of $prm\_leader$. Hence by the end of phase $2$ of the algorithm, there is exactly one DFS tree present in the graph.
\end{proof}

\begin{observation}\label{obs:leader}
   The leading agent with the largest value of $prm\_leader$ is elected as the leader among all the agents that are initially positioned arbitrarily at the nodes of the graph. 
\end{observation}

\begin{lemma}
    The execution of phase $3$ begins only after phase $2$ of the algorithm completes.
\end{lemma}
\begin{proof}
    The changes according to phase $3$ of the algorithm are made only after the agent $r_{st}$ updates $r_{st}.ph=3$. In other words, once the agent $r{st}$ has completed the execution of phase $2$ of the algorithm and has returned to its root node, it sets $r_{st}.ph=3$. Further, in the rounds $20-21$, the agents settled at an even distance from the root node in the DFS tree move to their parent node where $r$ is settled. The agent $r$ updates $r.crnt\_port$ only if $r.ph=3$ and $r_{st}$ is present with $r$. Thus, the changes according to phase $3$ occur only after phase $2$ is completed. 
\end{proof}

\begin{lemma}
    After the completion of phase $3$ of the algorithm, no two agents settle at neighboring nodes. Also, a maximal independent set of the graph is formed.
\end{lemma}
\begin{proof}
The agent $r_{st}$ moves through the tree edges formed in phase $2$ one by one (Algorithm \ref{alg: movement_r_st}). Let the agent $r_{st}$ reach a node $u$ where an agent $r_u$ is settled. The agent $r_u$ sets its $r_u.decision=1$ only if $r_u$ has no child (constructed in the DFS tree during phase $2$) whose decision is yet to be taken or $r_{st}$ wants to backtrack from the current node (lines $2$-$3$ and $10$ of Algorithm \ref{alg: movement_r_st}). When $r_u.decision$ is set to $1$, it moves through all the neighbors of $u$ to check if there is an agent that is already selected in the maximal independent set i.e. if there is an agent $r_i$ that has $r_i.final\_set=1$ (lines $2$-$13$ of Algorithm \ref{alg: settled_ph3}), $r_u$ is not included in the solution. Otherwise, it is included in the maximal independent set. Hence, the lemma follows. 
\end{proof}

\begin{lemma}
    The phase $3$ of the algorithm terminates correctly.
\end{lemma}
\begin{proof}
    According to the algorithm for a settled agent $r_j$ in phase $3$, if $r_j.decision=1$, it traverses through its ports one by one to check the status of its neighbors. The settled agents whose decision is yet to be taken, have $final\_set=0$, which represents that when they have $decision=1$, they will definitely traverse through the edge connecting $r_j$ and itself. Thus $r_j$ increments its value of $r_j.count$ by $1$ for such a settled agent in its neighborhood. Once the decision for $r_j$ is taken, it waits for its $count$ value to be $0$. As and when a settled agent $r_v$ with $r_v.decision=1$ visits $r_j$, $r_j$ decrements its value of $count$ by $1$. Thus, when the decision for each agent in the neighborhood of $r_j$ is taken, i.e., $r_j.count=0$, the agent $r_j$ terminates. The decision for each agent is taken one by one, thus each agent terminates when its value of $count$ becomes $0$. Finally, the agents settled at the root node traverse through its ports, and finally settle through their value of $r_{min}.final\_port$ and terminate. Thus, the phase $3$ of the algorithm terminates.
\end{proof}

\begin{theorem}\label{thm:d2d_main}
    Let $G$ be a port-labeled, undirected, and connected graph with $n$ nodes, $m$ edges, and maximum degree $\Delta$. Let $k(=n)$ agents be arbitrarily placed at it at $p$ multiplicity nodes and the agents have no prior knowledge of any of the global parameters $n$, $m$, $k$, $p$, and $\Delta$. The algorithm solves D-2-D with termination in $O(max\{n\log^2n, m\})$ rounds, and $O(\log n)$ memory is required by each agent to run the algorithm.
\end{theorem}
\begin{proof}
    As per the analysis mentioned in \cite{Sudo}, the dispersion requires $O(k(\log k)\cdot(\log min(k, \Delta))=O(k\log^2 k)$ time for execution. Thus, our phase $1$ of the algorithm completes in $O(n\log^2 n)$ time. The phase $2$ of the algorithm is run by the $leading$ agents. A $leading$ agent moves from a node $u$ to a node $v$ only if there is a settled agent present at node $v$. Thus, after phase $1$ is completed, the $leading$ agent with the largest value of $prm\_leader$ may take at most $m$ rounds to traverse through the entire graph and settle the $leading$ agents that could not complete their DFS traversal in phase $2$ at nodes for which the agent with $vprnt=1$ is acting as virtually settled. Finally, the $leading$ agent $r_{st}$ with the largest value of $prm\_leader$ initialises the phase $3$ of the algorithm by setting $r_{st}.ph=3$. It traverses through the tree edges of the DFS tree constructed in phase $2$ of the algorithm. Consequently, each edge is traversed twice. For a node $v$, the settled agent $r_v$ at $v$, whose decision is to be made, a total of $2\delta(v)$ rounds are required. Thus, at most $2n+m$ rounds are required to execute phase $3$ of the algorithm. Hence, our algorithm solves D-2-D with termination in $O(max\{n\log^2n, m\})$ rounds without any prior knowledge of any of the global parameters.      
    
    As described in the table \ref{tab:variables}, the variables $r_i.settled$, $r_i.state$, $r_i.ph$, $r_i.decision$, $r_i.dist$, $r_i.vprnt$, $r_i.sticky$, and $r_i.final\_set$ require $2$ bits of memory. The variables $r_i.recent$, $r_i.crnt\_port$, $r_i.prnt\_prt$, $r_i.prt\_ent$, $r_i.nxt\_prt$, $r_i.prm\_parent$, $r_i.final\_port$, $r_i.count$, $r_i.next$, and $r_i.v\_prnt\_port$ require $\log \Delta$ memory by the agents. The memory required by the agents to store the variable $r_i.ID$, $r_i.leader$, and $r_i.prm\_leader$ is $O(\log n)$. As $\Delta< n$, the algorithm can be executed with $O(\log n)$ bits of memory per agent.   
\end{proof}

\begin{corollary}
   Let $G$ be a port-labeled, undirected, and connected graph with $n$ nodes, $m$ edges, and maximum degree $\Delta$. Let $n$ agents be arbitrarily placed at the nodes of $G$ and have no prior knowledge of any of the global parameters. Our algorithm elects an agent as the leader in $O(max\{n\log^2n, m\})$ rounds, and $O(\log n)$ memory is required by each agent to run the algorithm.
\end{corollary}
\begin{proof}
    From Observation \ref{obs:leader}, it is clear that by the end of phase $2$, the agent with the largest value of $prm\_leader$ is elected as the leader from $n$ agents. From Theorem \ref{thm:d2d_main}, it is clear that D-2-D with termination is achieved in $O(\max\{n\log^2 n, m\})$ time. Thus, phase $2$ requires at most $O(\max\{n\log^2 n, m\})$ rounds for completion. The memory requirement of $O(\log n)$ with each agent for leader election holds true from Theorem \ref{thm:d2d_main}.
\end{proof}

\section{Conclusion}
In this work we provide an improved solution for D-2-D starting from any arbitrary initial configuration, using $n$ mobile agents, where $n$ is the number of nodes in the graph. It would be interesting to solve D-2-D from an arbitrary initial configuration with $k<n$ mobile agents.
An $\Omega(m\Delta)$ lower bound for D-2-D problem with $k>1$ agents is proved in \cite{kaur2023}. The proof uses a two-agent scenario. It would be interesting to do a lower-bound study of D-2-D starting with $k=O(n)$ agents. Also, solving D-2-D in the presence of faults can be another direction for further study.

\section{Acknowledgements}
Tanvir Kaur acknowledges the support from the CSIR, India (Grant No. 09/1005 (0048)/2020-EMR-I). Barun Gorain and Kaushik Mondal acknowledge the support by the Core Research Grant (no.: CRG/2020/005964) of the SERB, DST, Govt. of India. Barun Gorain acknowledges the support of the Science and Engineering Research Board (SERB), Department of Science and Technology, Govt. of India (Grant Number: MTR/2021/000118).
Also, this work was partially supported by the FIST program of the Department of Science and Technology, Government of India, Reference No. SR/FST/MS-I/2018/22(C).

\bibliographystyle{unsrt}

\end{document}